\documentclass[12pt,reqno]{amsart}
\usepackage{amsmath, amssymb, amsthm, amsfonts, amsxtra}

\usepackage{bbm}
\usepackage{mathrsfs}
\usepackage{amsfonts}
\usepackage{color}

\usepackage{amscd}

\usepackage{indentfirst,latexsym,bm}
\usepackage{graphics}
\usepackage{verbatim}

\usepackage[top=1.1 in, bottom=1.2 in, left=1.1 in, right=1.1 in]{geometry}

\newcommand{\nc}{\newcommand}

\def\1{^{-1}}

\newtheorem{thm}{Theorem}[section]

\newtheorem{cor}{Corollary}[section]

\newtheorem{exa}{Example}
\newtheorem{cla}{Claim}[section]

\theoremstyle{definition}

\newtheorem{rem}{Remark}[section]

\def\1{^{-1}}

\nc{\al}{\alpha} \nc{\bt}{\beta} \nc{\Gm}{\Gamma} \nc{\dl}{\delta}
\nc{\Dl}{\Delta} \nc{\lb}{\lambda} \nc{\Lb}{\Lambda}
\nc{\sg}{\sigma} \nc{\Sg}{\Sigma} \nc{\kp}{\kappa}
\nc{\vf}{\varphi} \nc{\ve}{\varepsilon}

\nc{\bB}{\Bbb{B}} \nc{\bC}{\Bbb{C}} \nc{\bR}{\Bbb{R}}
\nc{\bP}{\Bbb{P}} \nc{\bZ}{\Bbb{Z}}

\nc{\fR}{\frak{R}} \nc{\fF}{\frak{F}} \nc{\cG}{\cal{G}}
\nc{\cL}{\cal{L}} \nc{\cN}{\cal{N}} \nc{\cT}{\cal{T}}

\nc{\pa}{\partial} \nc{\wt}{\widetilde} \nc{\wh}{\widehat}
\nc{\ol}{\overline} \nc{\sbs}{\subset} \nc{\str}{\stackrel}
\nc{\bsl}{\setminus} \nc{\fa}{\forall} \nc{\DP}{\dot{+}}

\nc{\trace}{\operatorname{trace}} \nc{\RRe}{\operatorname{Re}}
\nc{\IIm}{\operatorname{Im}}

\title[Classification of singularities]{4d $\mathcal{N}=2$ SCFT and singularity theory Part III:  Rigid singularity}

\author{Bingyi Chen}
\address{Department of Mathematical Sciences,
Tsinghua University,
Beijing, 100084, P. R. China.}
\email{chenby16@mails.tsinghua.edu.cn}

\author{Dan \ Xie}
\address{Center of Mathematical Sciences and Applications\\
Jefferson Physical Laboratory, Harvard University, Cambridge, 02138, USA}
\email{dxie@cmsa.fas.harvard.edu}

\author{Stephen S.-T.\ Yau} 
\address{Department of Mathematical Sciences,
	Tsinghua University,
	Beijing, 100084, P. R. China.}
\email{yau@uic.edu}

\author{Shing-Tung\ Yau}
\address{Department of Mathematics \\
 Center of Mathematical Sciences and Applications\\
Jefferson Physical Laboratory, Harvard University, Cambridge, 02138, USA}
\email{yau@math.harvard.edu}

\author{Huaiqing Zuo}
\address{Yau Mathematical Sciences Center, Tsinghua University, Beijing, 100084, P. R. China.}
\email{hqzuo@mail.tsinghua.edu.cn}

\thanks{The work of S.-T. Yau is supported by  NSF grant  DMS-1159412, NSF grant PHY-
0937443, and NSF grant DMS-0804454. The work of Stephen S.-T. Yau is supported by NSFC grant 11531007 and Tsinghua University start-up fund. The work of Huaiqing Zuo is supported by NSFC (grant nos. 11771231, 11531007, 11401335) and Tsinghua University Initiative Scientic Research Program. 
The work of Dan Xie is supported by Center for Mathematical Sciences and Applications at Harvard University, and in part by the Fundamental Laws Initiative of the Center for the Fundamental Laws of Nature, Harvard University.}

\begin{document}

\maketitle

\begin{abstract}
  We classify three fold isolated quotient Gorenstein singularity $C^3/G$. These singularities are rigid, i.e. there is no non-trivial deformation, and  we conjecture that  they define 4d  $\mathcal{N}=2$  SCFTs which do not have a Coulomb branch.

	\end{abstract}
	
	\baselineskip=18pt
	\bigskip

\section{Introduction}
Four dimensional (4d) $\mathcal{N}=2$ superconformal field theory (SCFT) can be defined using type IIB string theory on  following background
\begin{equation}
 R^{1,3}\times X;
\end{equation}
Here $X$ is conjectured to be an isolated rational Gorenstein singularity \cite{XY} with a good $C^*$ action, and we take string coupling $g_s\rightarrow 0$ and go to infrared limit \cite{SV,GKP}.  These rational Gorenstein singularities naturally appear in the degeneration limit of compact Calabi-Yau three manifolds, and in fact general definition of Calabi-Yau variety allows such singularity \cite{G}. 

4d $\mathcal{N}=2$ SCFT has a $SU(2)_R\times U(1)_R$ R symmetry, and there are two kinds of half-BPS operators $ E_{r,(0,0)}$ and $ \hat{B}_1$ \cite{DO}. The Coulomb branch deformations are described as follows \cite{ALLM}:
\begin{enumerate}
\item Deformation using half-BPS operator $ E_{r,(0,0)}$:
\begin{equation}
\delta S = \lambda \int d^4x d Q^4 E_{r,(0,0)}+c.c.
\end{equation}
\item  Deformation using half-BPS operator $ \hat{B}_1$:
\begin{equation}
\delta S = m \int d^4x Q^2  \hat{B}_1+c.c.
\end{equation}
\item We can also turn on expectation value of operator $ E_{r,(0,0)}$: $ u_r=\langle E_{r,(0,0)} \rangle$. 
\end{enumerate}
A central question of understanding 4d $\mathcal{N}=2$ SCFT is to understand the low energy physics for general deformations parameterized by
$S=(\lambda, m, u_r)$. The low energy physics is best captured by the Seiberg-Witten geometry \cite{SW}. Usually Seiberg-Witten geometry is described by a family of Rieman surfaces fibered over space $S$, and it is conjectured in \cite{XY} that more general Coulomb branch geometry can be captured by the \textbf{mini-versal} deformation of certain kind of three fold singularity $X$ \cite{GLS}.  Roughly speaking, a deformation is a flat morphism $\pi: Y\rightarrow S$, with $\pi^{-1}(0)$ isomorphic to the singularity $X$, and a mini-versal deformation essentially captures all the deformations. Here $S$ is identified with the parameter space $(\lambda, m, u_r)$  of our (generalized) Coulomb branch. 

Therefore the study of 4d $\mathcal{N}=2$ SCFT and its Coulomb branch solution are reduced to the study of singularity $X$ and its mini-versal deformation. We have classified such $X$ which can be described by complete intersection \cite{XY,YY1,BD}, and the physical aspects of these 4d $\mathcal{N}=2$ SCFTs are studied in \cite{XY1,XY2,XY3,XYY}. All the complete intersection examples studied in \cite{XY,YY1,BD} have non-trivial mini-versal deformation and therefore 
 non-trivial  Coulomb branch.

The purpose of this note  is to study non-complete intersection rational Gorenstein singularities. An interesting class of such singularities are quotient singularity $C^3/G$ with $G$ a finite subgroup of $SL(3)$. One of main results of this paper is the classification of the three dimensional isolated Gorenstein quotient singularity.

We then would like to study mini-versal deformation of these singularities, and a surprising theorem by Schlessinger \cite{S} shows that all such singularities are rigid, i.e. they have no non-trivial deformation \footnote{See \cite{V} for example of rigid compact Calabi-Yau manifolds.}. Therefore the corresponding 4d theory has no Coulomb branch \footnote{Free hypermultiplets do have a Coulomb branch as we can turn on mass deformation.}. We call such theories rigid $\mathcal{N}=2$ theories. It would be very interesting to study more properties of these theories.

\section{Three-fold singularity and 4d $\mathcal{N}=2$ SCFT}
Let's discuss more about the interpretation of $\mathcal{N}=2$ SCFT defined using three fold rational Gorenstein singularity (they are also called canonical singularity \cite{R}). There are two special ways of smoothing a singularity: crepant resolution \cite{R} and mini-versal deformation \cite{GLS}. For the singularities we are interested, we have following facts:
\begin{itemize}
\item Every isolated singularity  has a mini-versal deformation \cite{GLS}, however, the deformation might be trivial.  A class of examples are the quotient singularity considered in this paper.
\item Every three fold canonical singularity has a crepant resolution $f:Y\rightarrow X$ such that $Y$ is Q-factorial \footnote{A Q-factorial variety means that  every Weil divisor on it is Q-Cartier, i.e., some multiple of it is a Cartier divisor.} \cite{K}. There is no  crepant resolution for Q factorial singularity. An example of Q-factorial singularity is the hypersurface singularity: $x^2+y^2+z^2+w^{2k+1}=0$ \cite{R}. The quotient singularity considered in this paper has a crepant resolution with $Y$ smooth as can be seen using toric method.
\item The only rational Gorenstein singularity that admits both trivial versal deformation and crepant resolution is the smooth point. 
\end{itemize} 

Now let's try to interpret the appearance of SCFT using the smoothing of singularity:
\begin{itemize}
\item If our singularity admits non-trivial deformation and the smooth manifold has three cycles (such as the hypersurface singularity), the low energy effective theory includes massless vector multiplet from compactifying self-dual RR four form, and we also have massive BPS states from D3 brane wrapping three cycles. These massive BPS states are in general mutually non-local. In the singular limit, the massive BPS states become massless, and it is expected that one get a SCFT \cite{APSW}.
\item If our singularity admits non-trivial crepant resolution, and the smooth manifold has two cycles and four cycles.  One can have massless hypermultiplets using various NS-NS and RR two forms, and one also have tensile strings from wrapping D3 branes on two cycles (or D5 branes on four cycles). In the singular limit, one get tensionless string and it is  expected that one get a SCFT \cite{W}.  
\end{itemize}
The SCFT considered in \cite{XY,WX} can be interpreted using the deformation of singularity, while the SCFT considered in this paper can be interpreted using crepant resolution. 

The Coulomb branch of a 4d theory is described by the deformation, while the Higgs branch is described by the crepant resolution. The exact Coulomb branch physics is described by the classical geometry of the deformation. The exact Higgs branch is difficult to compute, but we can count its dimension 
by computing the dimension of Mori cone \footnote{Mori cone describes the space of complete curves, which will generate free hypermultiplets.} associated with the crepant resolution. The number of abelian flavor symmetry is given by the rank of local 
class group of the singularity. 

\textbf{Example 1}: Let's consider a 3d singularity defined by equation $x^2+y^2+z^2+w^{2k+1}=0$, and the corresponding $\mathcal{N}=2$ SCFT is  $(A_1, A_{2k})$ Argyres-Douglas theory. The Coulomb branch is identified with the base of mini-versal deformation from which one can compute the Coulomb branch spectrum. There is no Higgs branch, and this agrees with the fact that there is no crepant resolution for the singularity.

\textbf{Example 2}: Let's consider the singularity $x^2+y^2+z^2+w^{2k}=0$, and the corresponding $\mathcal{N}=2$ SCFT is the $(A_1, A_{2k-1})$ Argyres-Douglas theory. The Coulomb branch is identified with the base of mini-versal deformation from which one can compute the Coulomb branch spectrum. There is a one dimensional Higgs branch, and this agrees with the fact that there is a crepant resolution whose Mori cone has dimension one!

\section{Classification of rigid quotient singularity}\label{sec1}

Let $G$ be a finite subgroup of $GL(3,\Bbb C)$ and it acts on $\Bbb C^3$ in a natural may. Cartan \cite{Car} has studied the quotient variety $\Bbb C^3/G$ and  proved that the singularities of $\Bbb C^3/G$ are normal. So the dimension of the singular set of $\Bbb C^3/G$ is either 0 or 1. In this article we are interested in the case that $\Bbb C^3/G$ has a  Gorenstein isolated singlarity. By a theorem of Khinich \cite{Kh} and Watanabe \cite{Wa}, we know that  
\begin{thm}(\cite{Kh} and \cite{Wa})
Let $G$ be a finite subgroup of $GL(3,\Bbb C)$. Then $\Bbb C^3/G$ is  Gorenstein  if and only if $G$ is a subgroup of $SL(3,\Bbb C)$.
\end{thm}
Let $G'$ be another finite subgroup of $GL(3,\Bbb C)$. We say $G$ is linear equivalent to $G'$ if there exists $g \in GL(3,\Bbb C)$ such that $G=gGg^{-1}$. It's obvious that $\Bbb C^3/G\cong \Bbb C^3/G'$ if $G$ is linear equivalent to $G'$. Yau and Yu [YY2] tell us that
\begin{thm}(\cite{YY2})\label{the0}
Let $G$ be a finite subgroup of $SL(3,\Bbb C)$, then $\Bbb C^3/G$ has a Gorenstein isolated singularity if and only if $G$ is linear equivalent to a diagonal abelian subgroup (i.e. any element in this subgroup is a diagonal matrix) and 1 is not an eigenvalue of $g$ for every nontrivial element $g$ in $G$.
\end{thm}
In this article, we will find out all subgroups $G\subseteq SL(3,\Bbb C)$ which satisfy the condition in Theorem \ref{the0}, i.e. all the subgroups which corresponds to a three-dimensional  Gorenstein isolated quotient singularity. In fact, we prove that
\begin{thm}\label{the1}
Let $G$ be a finite subgroup of $SL(3,\Bbb C)$. Then $\Bbb C^3/G$ has a Gorenstein isolated singularity if and only if $G$ is linear equivalent to a cyclic subgroup which is generated by a diagonal matrix 
$$\begin{pmatrix}
\zeta(1/n)& 0& 0\\
0& \zeta(p/n)& 0\\
0& 0& \zeta(q/n)
\end{pmatrix}$$
where $\zeta(\ast)=e^{2\pi \sqrt{-1}\ast}$ and  $p,q,n$ are positive integers such that $p,q$ are coprime with $n$ and $1+p+q=n$.
\end{thm}

A polynomial $f \in \Bbb C[x,y,z]$ is called an invariant polynomial of $G\subseteq SL(3,\Bbb C)$ if $f(g(p))=f(p)$ for any element $g \in G$ and any point $p \in \Bbb C^3 $. Denote by $S^G$ the subalgebra of $\Bbb C[x,y,z]$ that consists of all invariants of $G$. Then the quotient variety $\Bbb C^3/G$ is isomorphic to the algebraic variety Spec$(S^G)$. If $\{f_1,\dots,f_k\}$ is a minimal set of homogeneous polynomials which generated $S^G$ (as a $\Bbb C-$algebra), then we call $f_i's$ the minimal generators of $S^G$. Geometrically, $k$, the number of minimal generators of $S^G$, is the minimal embedding dimension of $\Bbb C^3/G$. 

Consider the following ring homomorphism
\begin{align*}
\phi:\Bbb C[y_1,\dots,y_k] &\rightarrow S^G\\
y_i\quad &\mapsto  f_i
\end{align*}
where $f_1,\dots,f_k$ are minimal generators of $S^G$.
Let $K$ be the kernel of $\phi$, then the generators of $K$ are called the relations of minimal generators $f_1,\dots,f_k$. Geometrically, these relations are the equations which define the affine variety Spec$(S^G)$ as a subvariety of $\Bbb C^k$. Associate to $y_1,y_2,\dots, y_k$ a weight system $(w_1,w_2,\dots,w_k)$, where
\begin{equation}\label{weight}
w_i=\deg f_i
\end{equation}
for $i=1,2,\dots,k$. With respect to this weight system, $K$ is a weighted homogeneous ideal of $\Bbb C[y_1,\dots,y_k]$, so $\Bbb C^3/G$ has a weighted homogeneous singularity.

In Section 3, we construct a set of minimal generators of $S^G$ and find out their relations for each subgroup $G$ corresponding to a  three-dimensional Gorenstein isolated quotient singularity. And we get the following corollary.

\begin{cor}
The minimal embedding dimension of a three-dimensional Gorenstein isolated quotient singularity $\Bbb C^3/G$ is no less than 10.
\end{cor}
\begin{rem}
\cite{BD} proves that the minimal embedding dimension of a three-dimensional rational isolated complete intersection singularity is at most 5. Hence a three-dimensional Gorenstein isolated quotient singularity must be non-complete intersection.
\end{rem}

At the end of this section, we introduce some notations.

(1) For any positive integer $k$, $k$ can be written as $$k=\frac{1}{a_1-\frac{1}{a_2-\frac{1}{...-\frac{1}{a_e}}}}$$
where $a_i's$ are  positive integers. It's called the continued fraction expansion of $k$, and is denoted by $$k=[a_1,a_2,\dots,a_e].$$
We call $e$ the length of the continued fraction expansion of $k$, which is denoted by $l(k)$.

(2) Let $g$ be a monomial in $\Bbb C[x_1,\dots,x_n]$.
Denote by Supp$(g)$ the set consists of variables involved in $g$. For example, if $g=x_1x_2$, then Supp$(g)=\{x_1,x_2\}$.

(3) We denote by $\langle a,b,c \rangle$ the $3 \times 3$ diagonal matrix whose diagonal elements are $a,b,c$. 
Similarly we denote by $\langle a,b \rangle$ the $2 \times 2$ diagonal matrix whose diagonal elements are $a,b$. 

(4) Let $\zeta(q)=e^{2 \pi \sqrt{-1}q}$ for any real  number $q$. 

(5) If $A$ is a matrix, we denote its $(i,j)$-entry by $A[i,j]$.

\begin{proof}[\textbf{Proof of Theorem \ref{the1}}]
We first prove the sufficiency. If $G$ is generated by a diagonal matrix $\langle \zeta(1/n),\zeta(p/n),\zeta(q/n)\rangle$, where $p,q,n$ are positive integers, $p,q$ are coprime with $n$ and $1+p+q=n$, then each element $g \in G$ can be written as $\langle \zeta(k/n),\zeta(kp/n),\zeta(kq/n)\rangle$ for some integer $k$. If 1 is an eigenvalue of $g$, since $1,p,q$ are coprime with $n$, we have $k \equiv 0$ (mod $n$), which follows that $g$ is the unit matrix. By Theorem \ref{the0}, $\Bbb C^3/G$ has a Gorenstein isolated singularity.

Next we prove the necessity. If $\Bbb C^3/G$ has an isolated singularity, then by Theorem \ref{the0}, 1 is not an eigenvalue of $g$ for every nontrivial element $g$ in $G$ and we may suppose that $G$ is a diagonal abelian subgroup. By the fundamental theorem for finite abelian groups, G is the direct sum of cyclic groups:
$$G=\oplus _{i=1}^m \oplus _{j=1}^{r_i} G_{ij} $$
where $G_{ij}$ is a cyclic group whose order is $p_i^{n_{ij}}$, $p_1,p_2,\dots,p_m$ are distinct prime numbers and 
$$1\leq n_{i1}\leq n_{i2}\leq\dots \leq n_{ir_i},\quad i=1,\dots,m.$$
$G_{ij}$ is generated by a diagonal matrix $g_{ij}=\langle \zeta(a_{ij}/p_i^{n_{ij}}),\zeta(b_{ij}/p_i^{n_{ij}}),\zeta(c_{ij}/p_i^{n_{ij}})\rangle$ for $i=1,\dots,m$ and $j=1,\dots,r_i$, and $a_{ij}+b_{ij}+c_{ij} \equiv 0$ (mod $p_i^{n_{ij}}$). Since $g_{ij}^t\neq I$ ($I$ is the unit matrix) for $1\leq t<p_i^{n_{ij}}$, 1 is not an eigenvalue of $g_{ij}^t$, hence $ta_{ij},tb_{ij},tc_{ij}\not \equiv 0 
$ (mod $p_i^{n_{ij}}$) for $1\leq  t<p_i^{n_{ij}}$. Thus $a_{ij},b_{ij},c_{ij}$ are coprime with $p_i$ for  $i=1,\dots,m$ and  $j=1,\dots,r_i$.

We claim that $r_i=1$ for $i=1,\dots,m$. Assume that $r_1>1$, and for convenience in the sequel we will denote $p_1$ by $p$ and denote $G_{1i},g_{1i},a_{1i},b_{1i},c_{1i},n_{1i}$ by $G_i,g_i,a_i,b_i,c_i,n_i$ respectively, so $G_1$ is generated by $g_1=\langle\zeta(a_1/p^{n_1}),\zeta(b_1/p^{n_1}),\zeta(c_1/p^{n_1})\rangle$ and $G_2$ is generated by $g_2=\langle\zeta(a_2/p^{n_2}),\zeta(b_2/p^{n_2}),\zeta(c_2/p^{n_2})\rangle$. Since $a_2$ is coprime with $p$, there exist a integer $s$ such that $p^{n_1} \mid (a_1+sa_2)$, hence $p^{n_2} \mid (p^{n_2-n_1}a_1+p^{n_2-n_1}sa_2) $. Let $s'=p^{n_2-n_1}s$ then $$\zeta(a_1/p^{n_1})\zeta(a_2/p^{n_2})^{s'}=\zeta((p^{n_2-n_1}a_1+s'a_2)/p^{n_2})=1,$$ thus 1 is an eigenvalue of $g_1g_2^{s'}$, which follows that $g_1g_2^{s'}=I$. Hence $$\zeta(b_1/p^{n_1})\zeta(b_2/p^{n_2})^{s'}=\zeta((p^{n_2-n_1}b_1+s'b_2)/p^{n_2})=1$$ and $$\zeta(c_1/p^{n_1})\zeta(c_2/p^{n_2})^{s'}=\zeta((p^{n_2-n_1}c_1+s'c_2)/p^{n_2})=1.$$
Thus $g_2^{s'}g_1=I$, which leads to contradiction with $G_1\cap G_2=\{I\}$. Thus $r_1=1$. Similarly we have $r_i=1$ for $i=1,2,\dots,m$ and thus $G$ is generated by matrices $g_i=\langle \zeta(a_{i}/p_i^{n_{i}}),\zeta(b_{i}/p_i^{n_{i}}),\zeta(c_{i}/p_i^{n_{i}})\rangle$ for $i=1, \cdots, m$, where $p_1,p_2,\cdots,p_m$ are distinct primes, $a_i,b_i,c_i$ are coprime with $p_i$ and $a_i+b_i+c_i\equiv 0$ (mod $p_i^{n_i}$). 

Since $a_i$ is coprime with $p_i$, there exists a  integer $s_i$ such that $0\leq s_i < p_i^{n_i}$ and $a_is_i+b_i\equiv 0 \text{ (mod } p_i^{n_i})$. Using the fact $p_i's$ are pairwise distinct prime and Chinese remainder theorem, there exist a integer $k$ such that $k \equiv s_i \text{ (mod } p_i^{n_i})$, hence $a_ik+b_i\equiv 0 \text{ (mod } p_i^{n_i})$. Let $n=\prod \limits_{i=1}^m p_i^{n_i}$. Next we prove that $G$ is generated by a matrix $$g=\langle \zeta(1/n),\zeta((n-k)/n),\zeta((k-1)/n)\rangle.$$ Let $k_i=a_j\prod \limits_{j\neq i} p_j^{n_j}$, then $k_i$ is coprime with $p_i$ for $i=1,2,\cdots m$. Then we have $g[1,1]^{k_i}=g_i[1,1]$ (as we have mentioned above $g[a,b]$ (resp. $g_i[a,b]$) means the $(a,b)$-entry of $g$ (resp. $g_i$)). And since $$g[2,2]=g[1,1]^{-k}, \quad \quad g_i[2,2]=g_i[1,1]^{-k}$$ $$g[3,3]=g[1,1]^{-1}g[2,2]^{-1}, \quad \quad g_i[3,3]=g_i[1,1]^{-1}g_i[2,2]^{-1},$$
we have $g[2,2]^{k_i}=g_i[2,2]$ and $g[3,3]^{k_i}=g_i[3,3]$, which follows that $g^{k_i}=g_i$. Since $k_i$ is coprime with $p_i$ for each $i$, then the greatest common divisor of $n,k_1,k_2,\cdots,k_n$ is 1, thus there exist $t_i$ such that $t_1k_1+t_2k_2+\cdots t_mk_m \equiv 1$ (mod $n$). Hence $\prod g_i^{t_i}=\prod g^{t_ik_i}=g$.  (because $g^n=1$). Hence $G$ is generated by the matrix $g$. 

Finally we only need to prove $n-k$ and $k-1$ is coprime with $n$. If $n-k$ is not coprime with $n$, then there exists $0<r<n$ such that $n \mid (n-k)r$. Then $g^r$ has eigenvalue 1 but $g^r$ is not the unit matrix, which leads to contradiction. Similarly we can prove that $k-1$ is coprime with $n$ and the theorem is proved.
\end{proof} 
\textbf{Minimal generators of the invariant ring and their relations}\label{sec3}:
Denote by $H_{n,p}$  the subgroup of $SL(3,\Bbb C)$ generated by the matrix $$ g_{n,p}=\langle \zeta(1/n),\zeta(p/n),\zeta((n-p-1)/n)\rangle$$ where $p$ and $n-p-1$ are coprime with $p$. By Theorem \ref{the1} we know that $\Bbb C^3/G$ defines a three-dimensional Gorenstein isolated singularity. A polynomial $f \in \Bbb C[x,y,z]$ is an invariant polynomial of $H_{n,p}$ if each term $x^ay^b z^c$ in $f$ satisfies
$$a+pb+(n-p-1)c\equiv 0 \text{ (mod } n).$$
Denote by $S_{n,p}$ the subalgebra of $\Bbb C[x,y,z]$ that consists of all invariants of $H_{n,p}$. Then $\Bbb C^3/H_{n,p}$ is isomorphic to the algebraic variety Spec$(S_{n,p})$. If $\{f_1,\dots,f_k\}$ is a minimal set of homogeneous polynomials such that $S_{n,p}$ is generated by $f_1,\cdots,f_k$ as a $\Bbb C-$algebra, then we call $f_i's$ minimal generators of $S_{n,p}$. Then the minimal embedding dimension of  $\Bbb C^3/H_{n,p}$ is equal to  the number of  minimal generators of $S_{n,p}$.

Let $f_1,\dots,f_k$ be minimal generators of $S_{n,p}$. Consider the ring homomorphism
\begin{align*}
\phi:\Bbb C[y_1,\dots,y_k] &\rightarrow S_{n,p}\\
y_i\quad &\mapsto  f_i
\end{align*}
Let $K_{n,p}$ be the kernel of $\phi$, the generators of $K_{n,p}$ (as an ideal of $\Bbb C[y_1,\dots,y_k]$) are called  relations of $f_1,f_2,\dots,f_k$. In this section we will determine a set of minimal generators of $S_{n,p}$ and find out their relations for all $n,p$ such that $p$ and $n-p-1$ are coprime with $n$.

First let's recall a result of Riemenschneider  \cite{R} about two-dimensional cyclic quotient singularities.

\begin{thm}\label{the2}(\cite{R})
Let $G=G_{n,p}$ be the subgroup of $SL(2,\Bbb C)$, generated by $\begin{pmatrix}
\zeta(1/n)&0\\
0&\zeta(p/n)
\end{pmatrix}$. The continue fraction of $n/(n-p)$ is $[a_1,a_2,\dots,a_e]$. Then a set of minimal generators of the invariant ring $\Bbb C[u,v]^G$ is $\{f_k=u^{i_k}v^{j_k}\}_{k=0}^{e+1}$, where $i_k,j_k$ are determined as follows:
\begin{equation}\label{eq1}
\begin{aligned}
&i_{0}=n,\quad i_{1}=n-p, \quad i_{k+1}=a_ki_k-i_{k-1}\quad \text{ for } 1\leq k\leq e\\
&j_0=0,\quad j_1=1, \quad \quad \quad j_{k+1}=a_kj_k-j_{k-1}\quad \text{ for } 1\leq k\leq e
\end{aligned}
\end{equation}
The relations of $\Bbb C[u,v]^G$ are 
\begin{align}\label{eq2}
f_{i-1}f_{j+1}=f_if_j\prod \limits_{k=i}^j f_k^{a_k-2}
\end{align}
for $0<i<j<e+1$.
\end{thm}
\begin{rem}
In fact  $i_{e+1}=0$ and $j_{e+1}=n$. So $f_0=u^n$ and $f_{e+1}=v^n$
\end{rem}
Let's see a example.
\begin{exa}\label{ex1}
Let $G=G_{3,1}$, then $3/(3-1)=[2,2]$ and $e=2$. We have
$$i_{0}=3,\quad i_{1}=2, \quad i_2=2i_1-i_0=1 \quad i_3=2i_2-i_1=0$$ 
$$j_0=0,\quad j_1=1, \quad j_2=2j_1-j_0=2 \quad j_3=2j_2-j_1=3$$ 
Thus $\Bbb C[u,v]^{G}$ is generated by 
$$\{f_0=u^3,f_1=u^2v,f_2=uv^2,f_3=v^3\}$$
And the relations are
$$\{f_0f_2=f_1^2,\quad f_0f_3=f_1f_2, \quad f_1f_3=f_2^2\}.$$
\end{exa}

Now come back to the three-dimensional case. Consider the subring $S_{n,p}\cap \Bbb C[x,y]$ of $S_{n,p}$. Since $S_{n,p}\cap \Bbb C[x,y]$ consists of all monomial $x^ay^b$ such that $a+p b\equiv 0 \text{ (mod } n)$, we have   $S_{n,p}\cap \Bbb C[x,y]=\Bbb C[x,y]^{G_{n,p}}$, where $G_{n,p}$ is the subgroup of $SL(2,\Bbb C)$ which is generated by $\langle \zeta(1/n),\zeta(p/n)\rangle$. Using Theorem \ref{the2}, we know that $S_{n,p}\cap \Bbb C[x,y]$ is generated by $\{f_{1,k}=x^{i_{1,k}}y^{j_{1,k}}\}_{k=0}^{e_1+1}$, where $e_1$ is the length of the continue fraction $n/(n-p)$, and $i_{1,k}$,$j_{1,k}$ is defined as equations (\ref{eq1}) in Theorem \ref{the2}. And the relations of $\{f_{1,k}=x^{i_{1,k}}y^{j_{1,k}}\}_{k=0}^{e_1+1}$ are 
\begin{align}\label{eq4}
f_{1,i-1}f_{1,j+1}=f_{1,i}f_{1,j}\prod \limits_{k=i}^j f_{1,k}^{a_{1,k}-2},
\end{align}
for $0<i<j<e_1+1$, where $[a_{1,1},a_{1,2},\dots,a_{1,e_1}]$ is the continue fraction of $n/(n-p)$. Denote the set $\{f_{1,k}=x^{i_{1,k}}y^{j_{1,k}}\}_{k=1}^{e_1}$ by $A_{xy}(n,p)$, then $S_{n,p}\cap \Bbb C[x,y]$ is generated by $A_{xy}(n,p)\cup\{x^n,y^n\}$. And we denote the set of  relations (\ref{eq4}) by $R_{xy}(n,p)$. Similarly,  $S_{n,p}\cap \Bbb C[x,z]=\Bbb C[x,z]^{G_{n,n-p-1}}$, and we denote the set of its minimal generators by $\{x^n,z^n\}\cup A_{xz}(n,n-p-1)=\{x^n,z^n,f_{2,1}=x^{i_{2,1}}z^{j_{2,1}},f_{2,2}=x^{i_{2,2}}z^{j_{2,2}},\dots,f_{2,e_2}=x^{i_{2,e_2}}z^{j_{2,e_2}}\}$ and denote the set  of relations  by $R_{xz}(n,n-p-1)$ . Next we consider $S_{n,p}\cap \Bbb C[y,z]$. Obviously  $S_{n,p}\cap \Bbb C[y,z]=\Bbb C[y,z]^{G}$ where $G$ is the subgroup of $SL(2,\Bbb C)$ generated by $\langle \zeta(p),\zeta(n-p-1) \rangle$. Since $p$ is coprime with $n$, there exist $q$ such that $pq \equiv 1$ (mod $n$) and $q$ is coprime with $n$. We have $q(n-p-1) \equiv r$ (mod $n$) for some positive integer $r$ less than $n$. Hence
$$\langle \zeta(p/n),\zeta((n-p-1)/n) \rangle^q=\langle \zeta(1/n),\zeta(r/n) \rangle$$
and $$\langle \zeta(p/n),\zeta((n-p-1)/n) \rangle=\langle \zeta(1/n),\zeta(r/n) \rangle^p.$$
Hence $G$ is generated by $\langle \zeta(1/n),\zeta(r/n)\rangle$. As before, we denote the set of minimal generator of $\Bbb C[y,z]^{G_{n,r}}$ by $\{y^n,z^n\} \cup A_{yz}(n,r)=\{y^n,z^n,f_{3,1}=y^{i_{3,1}}z^{j_{3,1}},f_{3,2}=y^{i_{3,2}}z^{j_{3,2}},\dots,f_{3,e_3}=y^{i_{3,e_3}}z^{j_{3,e_3}}\}$ and the set of their relations by $R_{yz}(n,r)$. Obviously $xyz \in S_{n,p}$, and our following theorem will prove that $\{g_1=x^n,g_2=y^n,g_3=z^n,g_4=xyz\} \cup A_{xy}(n,p)\cup A_{xz}(n,n-p-1) \cup A_{yz}(n,r)$ is a set of minimal generators of $S_{n,p}$. These  generators (exclude $g_4$) form a triangle as the following picture

\begin{equation}\label{ma1}
\begin{matrix}
g_1& & & & \\
f_{1,1}&f_{2,1} & & & \\
f_{1,2}& &f_{2,2} & &\\
\cdots & & &\cdots&  &\\
f_{1,e_1}& & & &f_{2,e_2} & \\
g_2&f_{3,1} &f_{3,2}& \cdots&f_{3,e_3}& g_3 
\end{matrix}
\end{equation}

We call $\{g_1,f_{1,1},f_{1,2},\dots,f_{1,e_1},g_2\},\{g_1,f_{2,1},\dots,f_{2,e_2},g_3\}$ and $\{g_2,f_{3,1},\dots,f_{3,e_3},g_3\}$ the first, second and third side of the triangle (\ref{ma1}) respectively. Relations of generators which lie on the same side of the above triangle have been known, now we need to explore relations of generators which are on different sides. Obverse that if we take two generators $f$ and $g$ which lie on different sides, for  example $g=g_1$ and $f=f_{3,1}$, then $g_4=xyz \mid fg$. Hence we  introduce the definition "basic form" of a element in $S_{n,p}$. For any monomial $h=x^ay^bz^c \in S_{n,p}$, without loss of generality, we may assume that $c=\min\{a,b,c\}$. Since $g_4=xyz \in S_{n,p}$, we have $x^{a-c}y^{b-c} \in S_{n,p} \cap \Bbb C[x,y]$, which follows that $x^{a-c}y^{b-c}$ can be generated by $\{g_1=x^n,g_2=y^n\}\cup A_{xy}(n,p)$. Hence $h=g_4^c \widetilde{h}(g_1,g_2,f_{1,1},\dots,f_{1,e_1})$ in $\Bbb C[x,y,z]$, where $\widetilde{h}(g_1,g_2,f_{1,1},\dots,f_{1,e_1})$ is a polynomial in $g_1,g_2,f_{1,1},\dots,f_{1,e_1}$. We call $g_4^c \widetilde{h}(g_1,g_2,f_{1,1},\dots,f_{1,e_1})$ a basic form of $h$, and denote it by $B(h)$. Similarly in other two cases ($a=\min\{a,b,c\}$ and $b=\min\{a,b,c\}$) we can define $B(h)$. 
Let's see an example for basic forms.
\begin{exa}\label{ex2}
Let $n=3$ and $p=1$. Then $S_{n,p}\cap \Bbb C[x,y]=\Bbb C[x,y]^{G_{3,1}}$, which is generated by $\{g_1=x^3,g_2=y^3\} \cup A_{xy}(3,1)$.
 From Example \ref{ex1} we know that 
$$A_{xy}(3,1)=\{f_{1,1}=x^2y,f_{1,2}=xy^2\}.$$
and 
$$R_{xy}(3,1)=\{g_1f_{1,2}=f_{1,1}^2,\quad g_1g_2=f_{1,1}f_{1,2}\quad f_{1,1}g_2=f_{1,2}^2\}.$$
Let $f=x^4y^4z \in S_{n,p}$, then $f=g_4\cdot x^3y^3$.
$x^3y^3 \in C[x,y]^{G_{3,1}}$ and it can be written as $f_{1,1}f_{1,2}$. Hence $B(f)=g_4f_{1,1}f_{1,2}$ is basic form of $f$.
\end{exa}
Now we can prove the main theorem of this section.
\begin{thm}\label{the3}
Using the notation above, then a set of minimal generators of the invariant ring $S_{n,p}$ is 
\begin{equation}\label{eq5}
\begin{aligned}
\{g_1=x^n,g_2=y^n,&g_3=z^n,g_4=xyz\}\cup \\
&A_{xy}(n,p)\cup A_{xz}(n,n-p-1)\cup A_{yz}(n,r).
\end{aligned}
\end{equation}
And the relations are 
$$R_{xy}(n,p)\cup R_{xz}(n,n-p-1)\cup R_{yz}(n,r)\cup $$
$$\{gf-B(gf)\mid \text{ generators } g,f \text{ do not lie on the same side of triangle (\ref{ma1}) }\}$$
where $B(gf)$ is a basic form of $gf$.
More explicitly, the relations are 
\begin{equation}\label{eq6}
\begin{aligned}
&R_{xy}(n,p)\cup R_{xz}(n,n-p-1)\cup R_{yz}(n,r)\cup \\
&\{g_1f-B(g_1f)\mid f \in A_{yz}(n,r)\} \cup \{g_2f-B(g_2f)\mid f \in A_{xz}(n,n-p-1)\}\cup\\
& \{g_3f-B(g_3f)\mid f \in A_{xy}(n,p)\} \cup\\
&\{fg-B(fg)\mid f\in A_{xy}(n,p),\quad g \in A_{xz}(n,n-p-1)\}\cup\\
&\{fg-B(fg)\mid f\in A_{xy}(n,p),\quad g \in A_{yz}(n,r)\}\cup \\
&\{fg-B(fg)\mid f\in A_{x z}(n,n-p-1),\quad g \in A_{y z}(n,r)\}.
\end{aligned}
\end{equation}
\end{thm}
\begin{rem}
It's easy to see that $\deg(gf)=\deg(B(gf))$ with respect to the weight system (\ref{weight}) for any $f,g\in \{g_1,\dots,g_4,f_{1,1},\dots, f_{1,e_1},f_{2,1},\dots,f_{2,e_2},f_{3,1},\dots,f_{3,e_3}\}$. Hence equations in (\ref{eq6}) are weighted homogeneous.
\end{rem}
\begin{proof}
For any element $f \in S_{n,p}$, from its basic form $B(f)$, we know that $f$ can be generated by $(\ref{eq5})$. Hence $(\ref{eq5})$ generate $S_{n,p}$. Theorem \ref{the2} tells us that $A_{xy}(n,p)\cup \{x^n,y^n\}$ are minimal generators of $S_{n,p}\cap \Bbb C[x,y]$, hence each element in $A_{xy}(n,p)\cup \{x^n,y^n\}$ can not be generated by other elements in $A_{xy}(n,p)\cup \{x^n,y^n\}$. Similarly each element in $A_{xz}(n,n-p-1)\cup \{x^n,z^n\}$ (resp. $A_{yz}(n,r)\cup \{y^n,z^n\}$) can not be generated by other elements in $A_{xz}(n,n-p-1)\cup \{x^n,z^n\}$ (resp. $A_{yz}(n,r)\cup \{y^n,z^n\}$). And it's clear that $xyz$ can not be generated by other elements in $(\ref{eq5})$. Hence $(\ref{eq5})$ are minimal generators.

Consider ring homomorphism
$$\phi:\Bbb C[g_1,\dots,g_4,f_{1,1},\dots, f_{1,e_1},f_{2,1},\dots,f_{2,e_2},f_{3,1},\dots,f_{3,e_3}]\rightarrow S_{n,p}$$
$$g_1\mapsto x^n \quad g_2\mapsto y^n \quad g_3\mapsto z^n \quad g_4\mapsto xyz$$
$$f_{1,k}\mapsto x^{i_{1,k}}y^{j_{1,k}} \quad 
f_{2,k}\mapsto x^{i_{2,k}}y^{j_{2,k}} \quad
f_{3,k}\mapsto x^{i_{3,k}}y^{j_{3,k}}$$
Denote the kernel of $\phi$ by $K_{n,p}$. We will prove that $K_{n,p}$ is generated by $(\ref{eq6})$ as an ideal of $\Bbb C[g_1,\dots,g_4,f_{1,1},\dots, f_{1,e_1},f_{2,1},\dots,f_{2,e_2},f_{3,1},\dots,f_{3,e_3}]$.

First let's prove a claim.
\begin{cla}\label{cla1}
Let $P=\Bbb C[g_1,\dots,g_4,f_{1,1},\dots, f_{1,e_1},f_{2,1},\dots,f_{2,e_2},f_{3,1},\dots,f_{3,e_3}]$.
For any monomial $F$ in $P$, there exists a non-negative integer $k$ and a monomial $H$ in $P$ such that 

(1) $F-g_4^kH$ is generated by (\ref{eq6}); 

(2) $H$ is independent of $g_4$;

(3) elements in Supp$(H)$ lie on a side of trangle (\ref{ma1}) and Supp$(H)$ contains at most one vertex of that side. (here Supp$(H)$ means the set consists of variables which appear in $H$).
More explicitly, this condition requires that $H$ satisfies one of the following conditions:

(i) Supp$(H)\subseteq \{g_1,f_{1,1},\dots,f_{1,e_1}\}$;

(ii) Supp$(H)\subseteq \{g_1,f_{2,1},\dots,f_{2,e_2}\}$;

(iii) Supp$(H)\subseteq \{g_2,f_{1,1},\dots,f_{1,e_1}\}$;

(iv) Supp$(H)\subseteq \{g_2,f_{3,1},\dots,f_{3,e_3}\}$;

(v) Supp$(H)\subseteq \{g_3,f_{2,1},\dots,f_{2,e_2}\}$;

(vi) Supp$(H)\subseteq \{g_3, f_{3,1},\dots,f_{3,e_3}\}$.
\end{cla}
\begin{proof}[Proof of Claim \ref{cla1}]
We prove this claim by induction on the weighted degree of $F$ (with respect to the weight system (\ref{weight})). Without loss of generality, we may assume that $F$ is independent of $g_4$ (if $F=g_4^kF'$, we can replace $F$ by $F'$). There are following three cases:

(a) There exist $g,f\in Supp(F)$ such that $f,g$ do not lie on the same side of the triangle (\ref{ma1}), then $F$ can be written as $$F=gfF'=B(gf)F'+(gf-B(gf))F'.$$ Because $\deg(gf)=\deg(B(gf))$ we have $\deg(F)=\deg(B(gf)F')$. Since $(gf-B(gf))F'$ is generated by (\ref{eq6}), we only need prove the claim for $B(gf)F'$. By the definition of $B(gf)$, we know that $g_4\mid B(gf)$. Hence $B(gf)F'$ can be written as $g_4F''$, then $\deg F''< \deg B(gf)F'=\deg F$. By inductive assumption, we know the claim holds for $F''$, which follows that the claim holds for $g_4F''=B(gf)F'$.

(b) $g_1g_2g_3 \mid F$. Write $F=g_1g_2g_3F'$. Since $g_1g_2=f_{1,1}f_{1,e_1}\prod_{k=1}^{e_1} f_{1,k}^{a_{1,k}-2} \in R_{x,y}(n,p)$, we only need to prove the claim for 
$f_{1,1}f_{1,e_1}\prod_{k=1}^{e_1} f_{1,k}^{a_{1,k}-2}g_3F'$, and this has been treated in case (a).

(c) Elements in Supp$(F)$  lie on the same  side of the triangle (\ref{ma1}). Without loss of generality, we may assume that $F$ is a monomial  on variables $g_1,g_2,f_{1,1},\dots,f_{1,e_1}$.
If $g_1g_2 \mid F$,  write $F=g_1^sg_2^tF'$, where $F'$ is independent of $g_1,g_2$ and we may suppose that $s\leq t$. Since $g_1g_2-f_{1,1}f_{1,e_1}\prod_{k=1}^{e_1} f_{1,k}^{a_{1,k}-2} \in R_{xy}(n,p)$, we have $$F-(f_{1,1}f_{1,e_1}\prod_{k=1}^{e_1} f_{1,k}^{a_{1,k}-2})^sg_2^{t-s}F'$$ can be generated by (\ref{eq6}). Let $H=(f_{1,1}f_{1,e_1}\prod_{k=1}^{e_1} f_{1,k}^{a_{1,k}-2})^sg_2^{t-s}F'$, then the claim holds.

\end{proof}

For any $F(g_1,\dots,g_4,f_{1,1},\dots f_{1,e_1},f_{2,1},\dots,f_{2,e_2},f_{3,1},\dots,f_{3,e_3})$ $\in K_{n,p}$, where $F$ is a polynomial in $4+e_1+e_2+e_3$ variables, then $\phi(F)=0$. By Claim \ref{cla1}, we may assume that $F=F_0+g_4F_1+g_4^2F_2+\dots+g_4^mF_m$, where $F_i$ is independent of $g_4$ and  each term of $F_i$ satisfies the condition (3) in Claim \ref{cla1}. Hence $xyz \nmid \phi(F_i)$ unless $\phi(F_i)= 0$.  Since $\phi(F)=0$, then we have 
$$\phi(F_0)+xyz\phi(F_1)+(xyz)^2\phi(F_2)+\dots+(xyz)^m\phi(F_m)=0$$
in $S_{n,p}$. Since $xyz \nmid \phi(F_i)$ unless $\phi(F_i)= 0$,  we have $\phi(F_i)=0$ for $i=0,1,\dots,m$. Now we only need to prove each $F_i$ can be generated by $(\ref{eq6})$. Since each term of $F_i$ satisfies the condition (3) in Claim \ref{cla1} and is independent of $g_4$, we can write $$F_i=H_1+H_2+H_3+c_0+c_1g_1^{k_1}+c_2g_2^{k_2}+c_3g_3^{k_3}$$ where $H_j$ is a polynomial such that each term $t$ in $H_j$ satisfies that 

(1) elements in Supp$(t)$ lie on the $j$-th side of the triangle (\ref{ma1}),

(2) Supp$(t)\cap\{f_{j,1},\dots,f_{j,e_j}\}\neq \emptyset$, 
\\for $j=1,2,3$. Then we have $xy \mid \phi(H_1)$,  $xz\mid \phi(H_2)$ and  $yz \mid \phi(H_3)$ and we have $\phi(H_1)\in \Bbb C[x,y]$, $\phi(H_2)\in \Bbb C[x,z]$, $\phi(H_3)\in \Bbb C[y,z]$.  Since $\phi(F_i)=\phi(H_1)+\phi(H_2)+\phi(H_3)+c_0+c_1x^{k_1n}+c_2y^{k_2n}+c_3z^{k_3n}=0$, and $xy \mid \phi(H_1)$ and $\phi(H_2)\in \Bbb C[x,z]$ and $\phi(H_3)\in \Bbb C[y,z]$, we get $\phi(H_1)=0$. Using Theorem \ref{the1}, we get $H_1$ is generated by $R_{xy}(n,p)$. Similarly $\phi(H_2)=\phi(H_3)=0$ and $H_2$ (resp. $H_3$) can be generated by $R_{xz}(n,n-p-1)$ (resp. $R_{yz}(n,r)$). Hence $c_0+c_1x^{k_1n}+c_2y^{k_2n}+c_3z^{k_3n}=0$, which follows that $c_0=c_1=c_2=c_3=0$. Hence $F_i=H_1+H_2+H_3$ can be generated by (\ref{eq6}), then the theorem is proved.

\end{proof}
\begin{cor}
The minimal embedding dimension $d$ of $\Bbb C^3/H_{n,k}$ is $$4+l(n/(n-p))+l(n/(p+1))+l(n/(n-r))\geq 10,$$
 where $l(k)$ means the length the continue fraction for a positive integer $k$. 
\end{cor}
\begin{proof}
Since the minimal embedding dimension $d$ of $\Bbb C^3/H_{n,k}$ is equal to the number of minimal generators, using Theorem \ref{the3}, we have 
$d=4+l(n/(n-p))+l(n/(p+1))+l(n/(n-r))$. And since $p,n-p-1$ and $r$ are coprime with $n$, we have $l(n/(n-p)),l(n/(p+1)),l(n/(n-r) \geq 2$. Hence
$$d\geq 10.$$.
\end{proof}

\begin{exa}
Let $H=H_{3,1}$ be the subgroup of $SL(3,\Bbb C)$ generated by $\langle \zeta(1/3),\zeta(1/3),\zeta(1/3)\rangle$. As in Example 2,  a set of minimal generators of $S_{3,1}$ are
$$g_1=x^3, \quad g_2=y^3,\quad g_3=z^3,\quad g_4=xyz, \quad f_{1,1}=x^2y,$$
$$f_{1,2}=xy^2, \quad f_{2,1}=x^2z,\quad f_{2,2}=xz^2,\quad f_{3,1}=y^2z, \quad f_{3,2}=yz^2.$$
And
$$R_{xy}(3,1)=\{g_1f_{1,2}=f_{1,1}^2,\quad g_1g_2=f_{1,1}f_{1,2}\quad f_{1,1}g_2=f_{1,2}^2\};$$
$$R_{xz}(3,1)=\{g_1f_{2,2}=f_{2,1}^2,\quad g_1g_3=f_{2,1}f_{2,2}\quad f_{2,1}g_3=f_{2,2}^2\};$$
$$R_{yz}(3,1)=\{g_2f_{3,2}=f_{3,1}^2,\quad g_2g_3=f_{3,1}f_{3,2}\quad f_{3,1}g_3=f_{3,2}^2\}.$$
And $\{gf-B(gf)\mid \text{ generators } g,f \text{ do not lie on the same side of triangle (\ref{ma1}) }\}$=
$$\{ g_1f_{3,1}=g_4f_{1,1}, \quad 
g_1f_{3,2}=g_4f_{2,1}, \quad g_2f_{2,1}=g_4f_{1,2}, \quad g_2f_{2,2}=g_4f_{3,1}, \quad g_3f_{1,1}=g_4f_{2,2},$$
$$g_3f_{1,2}=g_4f_{3,2},\quad
f_{1,1}f_{2,1}=g_1g_4, \quad f_{1,1}f_{2,2}=g_4f_{2,1}, \quad f_{1,1}f_{3,1}=g_4f_{1,2}, \quad f_{1,1}f_{3,2}=g_4^2,$$ 
$$ f_{1,2}f_{2,1}=g_4f_{1,1},
f_{1,2}f_{2,2}=g_4^2, \quad f_{1,2}f_{3,1}=g_2g_4, \quad f_{1,2}f_{3,2}=g_4f_{3,1}, \quad f_{2,1}f_{3,1}=g_4^2,$$
$$f_{2,1}f_{3,2}=g_4f_{2,2},\quad f_{2,2}f_{3,1}=g_4f_{3,2},\quad f_{2,2}f_{3,2}=g_3g_4\}.$$
\end{exa}
\section{Toric geometry perspective}
The cyclic quotient singularity is toric and we can use toric method to understand the examples studied above. 
Let's first review briefly the toric singularity, for more details, see \cite{CLS}. We start with a three dimensional standard lattice $N$, and 
its dual lattice $M$. A convex cone $\sigma$ in $N_R$ is defined by a set of lattice points $v_\rho$:
\begin{equation}
\sigma =\{r_1v_1+\ldots+r_nv_n,~r_i\geq 0\}.
\end{equation}
The dual cone is defined as 
\begin{equation}
\sigma^{\vee} = \{ m\cdot v_\rho \geq 0,~~m\in M_R\}.
\end{equation}
The toric singularity is defined as $Spec(\sigma^{\vee}\cap M)$. We have following facts:
\begin{itemize}
\item The Gorenstein condition implies that there is a lattice vector $m_0 \in M$ such that $m_0\cdot v_\rho =1$ for any vector $v_\rho$.
We can choose coordinate such that $v_\rho=(p_\rho,q_\rho,1)$, so a Gorenstein toric singularity is defined by a convex lattice polygon $P$. 
\item The isolated singularity implies that there is no internal lattice points on boundary edges of $P$. 
\item We are interested in the case where there is no flavor symmetry, and this implies that the local class group of the singularity is 
trivial. This implies that $P$ is a triangle. 
\end{itemize}
So we need to classify triangle $P$ with no lattice points on the boundary edges. Now we can put one vertex at origin using translational invariance, 
and we can also put another vertex at point $(1,0)$. The third vertex can be constrained so that  its coordinate is $(a,b)$ with $ a>0, b>0$. The constraints on $(a,b)$ so that 
there is no lattice point on boundary edges are 
\begin{equation}
(a,b)=1,~~(a-1,b)=1
\label{tr}
\end{equation}
Here $(p,q)$ means the maximal common divisor of $p~and~q$.  See figure. \ref{toric} for the example. 

Now let's compare our result with theorem \ref{the1}, where the defining data also involves two positive integers $n, p$ such that 
$(n,p)=1$ and $(n, n-p-1)=1$, with $0<p<n$. With some computation, one can see that the classification from toric perspective is 
the same as that from the quotient singularity point of view.

\begin{figure}[h]
\centering
  \includegraphics{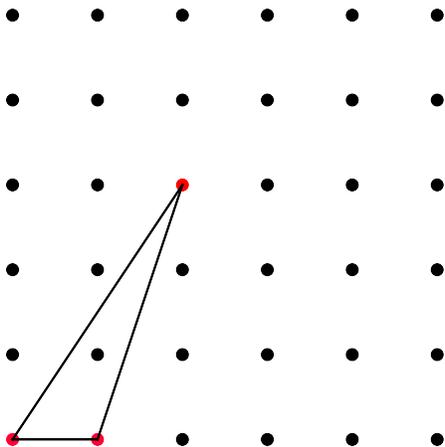}
  \caption{Isolated toric Gorenstein singularity with trivial class group is defined by a lattice triangle with no lattice points on the boundary.  }
  \label{toric}
  \end{figure}

Finally, we would like to point out that the deformation theory of isolated Gorenstein toric singularity has been studied in \cite{AL}, and above singularity 
is indeed rigid. The crepant resolution of the singularity is found from the unimodular lattice triangulation of $P$, from which we can read off the Higgs branch dimension.

\section{Discussion}
The singularities studied in this paper has trivial mini-versal deformation, and the underlying 
four dimensional $\mathcal{N}=2$ SCFT has no Coulomb branch (including mass deformation). The singularity admits 
non-trivial crepant resolution, and so it should have non-trivial Higgs branch. For example, 
$C^3/Z_3$ singularity has a crepant resolution with one exceptional divisor which is nothing but a 
$CP^2$.  There is one compact curve on resolved geometry and we expect  the Higgs branch 
to be one dimensional. This theory should have no flavor symmetry, since otherwise one can turn on 
mass deformation and then have non-trivial Coulomb branch. This fact is verified from toric point of view as the local class group is trivial.
While there are many 4d $\mathcal{N}=2$ theories admitting
no Higgs branch, to our knowledge we do not know any example admitting no Coulomb branch. 

From Higgs branch point view, the SCFT point is nontrivial as there are already massless degree of freedom in the deformed theory. The question is wether they 
are just free hypermultiplets. We used tensionless string argument to argue that the theory is interacting. Another reasoning 
is that if the theory is free,  we should see the flavor symmetry and the mass deformation which are all absent in the geometry. 
Given these reasonings, we tend to believe that the theory is interacting.   
We believe that  examples presented in this paper can help us better understand the space of  4d $\mathcal{N}=2$ SCFTs.

\providecommand{\bysame}{\leavevmode\hbox to3em{\hrulefill}\thinspace}
\providecommand{\MR}{\relax\ifhmode\unskip\space\fi MR }
\providecommand{\MRhref}[2]{%
  \href{http://www.ams.org/mathscinet-getitem?mr=#1}{#2}
}
\providecommand{\href}[2]{#2}

\end{document}